\documentclass[pra,aps,showpacs,amsfonts,superscriptaddress,twocolumn,nofootinbib]{revtex4-2}
\usepackage{color}
\definecolor{myurlcolor}{rgb}{0,0,0.4}
\definecolor{mycitecolor}{rgb}{0,0.5,0}
\definecolor{myrefcolor}{rgb}{0.5,0,0}
\usepackage{hyperref}
\hypersetup{colorlinks,
linkcolor=myrefcolor,
citecolor=mycitecolor,
urlcolor=myurlcolor}
\usepackage[mathscr]{euscript}
\usepackage[draft]{fixme}
\usepackage{tikz}
\usepackage{subfigure}
\usepackage{graphicx}
\usepackage{amsfonts}
\usepackage{amsmath, amssymb, amsthm,verbatim,graphicx,bbm}
\usepackage{palatino}
\usepackage{mathpazo}
\usepackage{dsfont}
\usetikzlibrary{patterns}
\usepackage{soul}
\newcommand{\beq}[0]{\begin{equation}}
\newcommand{\eeq}[0]{\end{equation}}

\newcommand{\outerp}[2]{\ket{#1}\! \bra{#2}}

\newcommand{\ket}[1]{|#1\rangle}
\newcommand{\bra}[1]{\langle#1|}

\newcommand{\one}{\leavevmode\hbox{\small1\normalsize\kern-.33em1}}

\def\be{\begin{equation}}
\def\ee{\end{equation}}
\def\ben{\begin{eqnarray}}
\def\een{\end{eqnarray}}
\def\eea{\end{array}}
\def\bea{\begin{array}}

\newcommand{\bei}{\begin{itemize}}
\newcommand{\eei}{\end{itemize}}

\renewcommand{\emph}[1]{\textbf{#1}}

\newcommand{\eqnref}[1]{(\ref{#1})}

\theoremstyle{plain}
\newtheorem{thm}{Theorem}

\newtheorem{lem}{Lemma}
\newtheorem{fakt}{Fact}
\newtheorem{prop}[thm]{Proposition}
\newtheorem{cor}[thm]{Corollary}

\newtheorem{defn}[thm]{Definition}

\theoremstyle{definition}

\theoremstyle{remark}

\begin{document}
\title{Completely entangled subspaces of entanglement depth $k$}
\author{Maciej Demianowicz}
\affiliation{Institute of Physics and Applied Computer Science, Faculty of Applied Physics and Mathematics, Gda\'nsk University of Technology,
Narutowicza 11/12, 80-233 Gda\'nsk, Poland}
\author{Kajetan Vogtt}
\affiliation{Center for Theoretical Physics, Polish Academy of Sciences, Aleja Lotnik\'ow 32/46, 02-668 Warsaw, Poland}
\affiliation{Faculty of Physics, University of Warsaw, Pasteura 5, 02-093 Warsaw, Poland}
\author{Remigiusz Augusiak}
\affiliation{Center for Theoretical Physics, Polish Academy of Sciences, Aleja Lotnik\'ow 32/46, 02-668 Warsaw, Poland}

\date{\today}

\begin{abstract}
   We introduce a class of entangled subspaces: completely entangled subspaces of entanglement depth $k$ ($k$-CESs). These are subspaces of multipartite Hilbert spaces containing only pure states with an entanglement depth of at least $k$. We present an efficient construction of $k$-CESs of any achievable dimensionality in any multipartite scenario. Further, we discuss the relation between these subspaces and unextendible product bases (UPBs). In particular, we establish that there is a non-trivial bound on the cardinality of a UPB whose orthocomplement is a $k$-CES. Further, we discuss the existence  of such UPBs  for qubit systems.
\end{abstract}

\maketitle
\section{Introduction}
Entanglement has long been known to be an enabling resource for a variety of  
information processing protocols such as for instance quantum communication \cite{144km} or quantum metrology \cite{quantum-metrology0,quantum-metrology}.
The ongoing rapid development of quantum network technologies requires a thorough theoretical characterization of its multipartite facet. It is a very challenging problem as the description of quantum systems with multiple nodes is far more complex than that of systems with only two subsystems \cite{three-qubits,four-qubits}. The topic has attracted much attention over the recent years and many profound results touching its different aspects have been reported in the literature (see, e.g., \cite{be-activation,negative,black-hole,tensor-rank,causal,activation-gme}). 
Noteworthy, the theoretical advancements go largely in parallel with a remarkable progress in the experimental domain (see, e.g., \cite{multiphoton,exp-multipartite,exp1,exp2,exp3}). Still, despite all the efforts, many aspects of many-body entanglement remain insufficiently explored and further research to better understand it is highly necessary.  This work is in line with this important trend.

A particular useful tool for the characterization of entanglement in multipartite systems is the entanglement depth \cite{depth}, which indicates how many particles in a given ensemble are genuinely entangled. The concept has found applications for example in the domain of cold gases \cite{bose1,bose2,bose3}.
The notion of entanglement can also be meaningfully considered for subspaces in addition to individual states. The most general notion is that of completely entangled subspaces (CESs), which are those subspaces that contain only pure entangled states 
\cite{wallach,partha,bhat,walgate}.  More specific classes of CESs considered in the literature are, e.g., subspaces containing only states with bounded tensor rank ($r$-entangled subspaces) \cite{cubitt,johnston-r}, or
genuinely entangled subspaces \cite{partha,upb-to-ges}, which are particularly interesting in the multipartite scenario as they contain only genuinely entangled states---the most resourceful states in this framework.
Interesting examples of genuinely entangled subspaces are for instance those spanned by the stabilizer error correction codes, including
the five-qubit \cite{5qubit} and toric  \cite{Kitaev} codes.

Entangled subspaces play an important role in quantum information science making them objects of intrinsic interest (see, e.g., \cite{makuta,john-hierarchia,john-ogolna}, for some recent results). First of all, they constitute an invaluable tool in the theory of entanglement because they can be readily utilized for a construction of mixed entangled states owing to the simple fact that any state supported on an entangled subspace is necessarily entangled.  
An important class of such constructed states comprises those built from unextendible product bases \cite{upb} and having positive partial transpose (PPT). These states share the appealing property of being bound entangled, i.e., their entanglement cannot be distilled into singlets with local operations and classical communication (LOCC). In fact, this was the first general construction of multipartite bound entangled states harnessing the notion of completely entangled subspaces.
Second, it was shown that almost all subspaces are completely entangled provided they are not too large, i.e., their dimension does not exceed the maximal permissible one for CESs  \cite{walgate}.
This result was utilized to prove that almost all sets of states 
(separable or entangled) are
locally unambiguously distinguishable  if the number of elements in the set is not larger than the difference between the dimension of the whole system and the maximal dimension of a CES in the given setup \cite{chefles, walgate}.
Moreover, entangled subspaces, such as the antisymmetric one and that introduced in Ref. \cite{partha}, have been exploited to provide  counterexamples to the additivity of the minimum output R\'enyi entropy of quantum channels \cite{grudka,studzinski}.
Their practical applicability further adds to their significance.
It has been recognized that they are relevant in the dynamically growing field of quantum error correction \cite{gour,scott}. A recent development in the research exploring this connection revealed that there exists a certain trade-off relating the larger capability of a code to correct errors to higher entanglement of the code space \cite{huber}. Furthermore, a link between entangled subspaces and self-testing was established \cite{max-nonlocal,dev-ind-ges,Owidiusz,Frerot} pointing to their use, e.g., in cryptography.

Our current contribution is to marry the notions of the entanglement depth and entangled subspaces.
Precisely, we introduce completely entangled subspaces of entanglement depth $k$ ($k$-CESs), as those subspaces containing only pure states whose entanglement depth is at least $k$. We find the maximal dimensions of such subspaces and present their universal construction  from sets of non-orthogonal product vectors in any multipartite setup. Further, we  discuss the relation between $k$-CESs and unextendible product bases (UPBs). This includes the derivation of a non-trivial bound on the cardinalities of UPBs leading to $k$-CESs and a discussion with a positive conclusion about their existence.

\section{Preliminaries}\label{preliminaries}
In this section we introduce  the relevant notions and terminology.

We consider quantum states defined on  $n$-partite ($n\ge 3 $) Hilbert spaces 
\begin{equation}
    \mathcal{H}_{d_1\dots d_n} :=\bigotimes_{i=1}^n \mathcal{H}_i= 
    \mathbb{C}^{d_1} \otimes \mathbb{C}^{d_2}\otimes \ldots \otimes \mathbb{C}^{d_n}.
\end{equation}
 We will mainly focus on the scenario with equal local dimensions $d_i=d$ in which case the space will be denoted as $\mathcal{H}_{d^n}$.  Single-partite subsystems (parties) are denoted $A_1,A_2,\dots, A_n:=\textbf{A}$. 
 A division of the set of parties into $K$ non-overlapping non-empty sets $S_i$ such that $\bigcup_i S_i=\textbf{A}$ is called a $K$-partition and is denoted $S_1|\dots|S_K$.  

An $n$-partite state  $\ket{\Phi}  \in \mathcal{H}_{d_1\dots d_n}$ is called {\it $k$-producible} if it can be written as 
\begin{equation}
\ket{\Phi}=\ket{\phi_1}\otimes \ldots\otimes \ket{\phi_M},
\end{equation}
where $\ket{\phi_i}$'s are at most $k$-partite states.   
If a $k$-producible state is not $(k-1)$-producible at the same time it is said to have {\it entanglement depth $k$} \cite{depth,producible}. States with entanglement depth $k=1$ are named {\it fully product} because they are products of pure states on individual subsystems $A_i$. All other states with $k \ge 2$ are said to be entangled and those with entanglement depth equal $n$, in particular, are called {\it genuinely multipartite entangled (GME)}. These are the states which cannot be written as a product for any bipartition of the parties. 
Excellent examples of GME states are the Greenberger-Horne-Zeilinger (GHZ) states \cite{GHZ}, 
\begin{equation}\label{GHZ}
    \ket{\mathrm{GHZ}_n}=\frac{1}{\sqrt{2}}(\ket{0}^{\otimes n}+\ket{1}^{\otimes n}),
\end{equation}
which belong to a general class of GME states called graph states  \cite{graphstates}.

Let us now recall the concept of entangled subspaces. A subspace of a multipartite Hilbert space $\mathcal{H}_{d_1\dots d_n}$ is called \textit{completely entangled (CES)} if all pure states belonging to it are entangled \cite{wallach,partha}. Importantly, entanglement of these states can be of any kind, including  only bipartite. In the extreme case of all states from a subspace being  GME one deals with genuinely entangled subspaces (GESs) \cite{upb-to-ges}. This is the case for instance for the antisymmetric subspace.

Entangled subspaces are closely related to the so-called unextendible product bases and we will extensively exploit this connection in later parts of the paper. Below we remind the basic necessary facts about the discussed notions.

  Consider 
  a set of fully product mutually orthogonal $n$-partite vectors from $\mathcal{H}_{d_1\dots d_n}$
\begin{equation}\label{dupa}
 \mathcal{B}:=\left\{\ket{\psi^{(j)}}=\bigotimes_{i=1}^n \ket{\varphi_i^{(j)}} \right\}_{j=1}^m, 
\end{equation}
where $\ket{\varphi_i^{(j)}} \in \mathbb{C}^{d_i}$.  
 $\mathcal{B}$ 
 is called an \textit{unextendible product basis (UPB)} iff it does not span the total Hilbert space $\mathcal{H}_{d_1\dots d_n}$ and there does not exist a fully product vector in the subspace complementary to $\mathrm{span}\:\mathcal{B}$ \cite{upb,big-upb}.  We will refer to the number of elements in $\mathcal{B}$ as the size or cardinality of it.

Clearly, by the very definition of a UPB, its complementary subspace is a CES.
To illustrate this general link with a simple example let us consider 
the following set of three-qubit fully product vectors 
 introduced in Ref. \cite{upb}:
%
\begin{equation}\label{Shifts}
\mathcal{S}:=\{\ket{000},\ket{1+-},\ket{-1+},\ket{+-1}\},
\end{equation}
where $\ket{\pm}=(\ket{0}\pm\ket{1})/\sqrt{2}$. One can easily verify that $\mathcal{S}$ is indeed a UPB and only entangled vectors belong to $(\mathrm{span} \:\mathcal{S})^{\perp}$, which itself is thus a CES. Notably, the latter subspace is not a GES because it contains biproduct vectors, e.g., $\ket{1}\otimes \ket{\phi}$, where $\ket{\phi}$ is an entangled two-qubit vector orthogonal to $\ket{+-}$, $\ket{1+}$, and $\ket{-1}$.

Concluding this section, we give two results 
regarding the unextendibility of a set of product vectors that will be our main tools in later parts. The first is the following observation.
\begin{fakt}\label{crucial} \cite{upb,big-upb}
Consider a set of $n$-partite product vectors $\mathcal{B}$ defined in Eq. (\ref{dupa}). 
%
%
There exists a fully product vector orthogonal to $\mathcal{B}$
if and only if there exists a partition of $\mathcal{B}$ into $n$ disjoint subsets: $\mathcal{B}=\mathcal{B}_1 \cup \dots \cup \mathcal{B}_n$, such that for all subsets $\mathcal{B}_i$ ($i=1,2,\dots,n$), the local vectors $\{\ket{\varphi_{i}^{(j)}}: \ket{\psi^{(j)}}\in \mathcal{B}_i\}$ do not span the corresponding local Hilbert spaces $\mathcal{H}_i=\mathbb{C}^{d_i}$.
\end{fakt}

This elegant result often serves as a basis for the constructions of UPBs and it will be used by us for a direct check of whether one of the considered bases possesses certain property of unextendibility. Its immediate consequence  is a lower bound on the cardinalities of unextendible sets:
\begin{equation}\label{condition}
 m \geq \sum_{i=1}^n(d_i-1) + 1.
\end{equation}

A stronger version of Fact \ref{crucial} follows as a corollary. 

\begin{cor}\label{corollary}
Consider a set of product vectors $\mathcal{B}$ defined in 
Eq. (\ref{dupa}) satisfying the condition \eqref{condition}. 
If for every $i$  any $d_i$-tuple of local vectors  $\{\ket{\varphi_i^{(j)}}\}_j$ spans the corresponding local Hilbert space $\mathcal{H}_i=\mathbb{C}^{d_i}$, then the set $\mathcal{B}$ is unextendible with product vectors. 
\end{cor}

In fact, for sets with the minimal cardinality $\sum_i(d_i-1)+1$ the constituent vectors must necessarily have the property of full local spanning on subsystems given in the corollary above.

Importantly, 
both results related to unextendibility, Fact \ref{crucial} and Corollary \ref{corollary}, are also applicable to sets of non-orthogonal vectors.

\section{Completely entangled subspaces of entanglement depth $k$ ($k$-CES)}\label{ka-cesy}

Let us move to the main body of the paper.
We introduce a class of entangled subspaces defined through the entanglement depth of vectors belonging to them. We propose the following.

 \begin{defn}
A subspace $S \subset \mathcal{H}_{d_1\dots d_n}$ is called a completely entangled subspace of entanglement depth $k$ ($k$-CES) if the entanglement depth of any vector belonging to $\mathcal{S}$ is at least $k$. Equivalently, a $k$-CES is void of $(k-1)$-producible states.
 \end{defn}

Let us remark that for $k=2$ the above definition 
recovers the definition of  completely entangled subspaces, whereas for $k=n$ it recovers the definition of genuinely entangled subspaces.

Before delving into a characterization of $k$-CESs, let us give a few simple examples of them, which will allow for a quick grasp of the notion. 

As the first example let us take the above-discussed 
CES complementary to the UPB given in Eq. \eqref{Shifts}. 
Within our  terminology 
we can now say that it is a $2$-CES (CES), but it is not a $3$-CES (GES) because, as discussed, the complementary subspace contains a three-qubit vector with an entanglement depth of $2$.

For the second example, we turn our attention to quantum error correction.
Consider the celebrated nine-qubit Shor's code \cite{shor} described by the vectors
\begin{align}
    \ket{\psi_0}&=\ket{\mathrm{GHZ}_3}\otimes\ket{\mathrm{GHZ}_3}
    \otimes \ket{\mathrm{GHZ}_3}, \\
%
%
    \ket{\psi_1}&=\ket{\overline{\mathrm{GHZ}}_3}\otimes\ket{\overline{\mathrm{GHZ}}_3}
    \otimes \ket{\overline{\mathrm{GHZ}}_3},
\end{align}
where $\ket{\overline{\mathrm{GHZ}}_3}=(1/\sqrt{2})(\ket{000}-\ket{111})$ and
$\ket{\mathrm{GHZ}_3}$ is defined in Eq. (\ref{GHZ}). Now, the subspace $\mathrm{span}\{\ket{\psi_0},\ket{\psi_1}\}$ is a $3$-CES because the entanglement depth of both $\ket{\psi_0}$ and $\ket{\psi_1}$ is $3$, and, at the same time, 
any linear combination of these vectors is GME, i.e., it has an entanglement depth equal to $9$.

As the final example we give an elementary general construction. 
 With this aim take a GES of a $k$-partite Hilbert space. Now, take a tensor product of all the vectors from this subspace with an arbitrary, possibly different for each vector, fully product $(n-k)$-partite vector. This clearly results in a $k$-CES of an $n$-partite Hilbert space.
This rather trivial method does not lead to $k$-CESs of the maximal dimensions for $k<n$ as can be seen from dimension consideration (see Appendix \ref{monotonicznosc-D}) and in Section \ref{konstrukcja} we will provide a more elaborate construction which does this job. However, it is noteworthy that if we choose the otherwise arbitrary vectors to be the same for each of the vectors from a GES, the resulting subspace will be a $k$-CES with all vectors belonging to it having an entanglement depth equal exactly to $k$.

One of the central problems of the theory of entangled subspaces is the problem of verifying whether a given subspace is entangled and further determining  to which class it belongs. It is worth noting that some tools have been developed in the literature so far, which could also be applied for the case of $k$-CESs. Specifically, the simple criterion based on the entanglement of basis vectors from Ref. \cite{bound-njp} is directly applicable and the hierarchical method from Ref. \cite{john-hierarchia} can be readily adapted. Another vital problem in the area is the quantification of subspace entanglement (see \cite{new-ent-sub} for recent results). Here, also, we find certain methods which can be used in the current case, e.g., the semidefinite programming bounds of Ref. \cite{ent-of-ges}.

Let us now investigate how big---in terms of their dimensions---$k$-CESs can be. 
 The theory of completely entangled subspaces 
helps us address this question successfully.  
For simplicity, we concentrate on the case of equal local dimensions in our derivation but a properly adjusted argument is 
applicable in the general case too. We have the following proposition.

\begin{prop}
The maximal attainable dimension of a $k$-CES
of $\mathcal{H}_{d^n}$ equals
\begin{align}
\label{kces}
D_{k-CES}^{max} =  d^n-\left(t d^{k-1} + d^{n-(k-1)t}-t\right),
\end{align}
where
\begin{equation}\label{t}
t=\left\lfloor \frac{n}{k-1} \right\rfloor.
\end{equation}
\end{prop}
\begin{proof}
By definition, a $k$-CES $\mathcal{S} \subset \mathcal{H}_{d^n}$ does not contain $(k-1)$-producible pure states. This implies that $\mathcal{S}$ is a CES in any $r$-partite Hilbert space
\begin{equation}\label{hilbert-new}
    (\mathbb{C}^d )^{\otimes n_1} \otimes (\mathbb{C}^d )^{\otimes n_2} \otimes \ldots \otimes (\mathbb{C}^d )^{\otimes n_r}, 
\end{equation}
%
where $n_i \le k-1$ and  $i=1,2,\dots,r$, 
corresponding to an $r$-partition $S_1 | S_2|\dots |S_r$ of the parties, where the size of each subset $S_i$ is $|S_i|=n_i$.

It was shown in Refs. \cite{wallach,partha} that the maximal attainable dimension of a CES of $\mathcal{H}_{d_1\dots d_n}$
is given by 
%
$\prod_{i=1}^{n} d_i - [ \sum_{i=1}^n (d_i-1) + 1$] (notice that the term in the square bracket happens to match the right-hand side of Eq. (\ref{condition})).
%
Applying this result to \eqref{hilbert-new}, we obtain
\begin{equation} \label{do-minimalizacji}
 d^n-\left[\sum_{i=1}^{r}(d^{n_i}-1) + 1 \right],
\end{equation}
%
which must further be minimized (equivalently, the term in the square brackets must be maximized) over $n_i$'s to obtain an upper bound on the maximal dimension of a $k$-CES. This is achieved for $n_i=k-1$, $i=1,2,\dots,t$, and $n_{t+1}=n-t(k-1)$, where $t=\left\lfloor n/(k-1)\right\rfloor$ (see Appendix \ref{max-dim}). The minimal value of \eqref{do-minimalizacji} is thus 
\begin{align}
d^n-\left(t d^{k-1} + d^{n-(k-1)t}-t\right).
\end{align}
The choice above corresponds to the partitions of the parties as follows (we assume that for $k-1 | n$ we have $t$-partitions)
\begin{equation} \label{new-hilbert}
    (\mathbb{C}^{d^{k-1}} )^{\otimes t} \otimes \mathbb{C}^{d^{n-(k-1)t}}.
\end{equation}

Further, generic subspaces are known to saturate bounds on the dimensions of entangled subspaces (cf. \cite{cubitt}), meaning that a random subspace of a proper dimension will be a $k$-CES. 
\end{proof}

The maximal dimension \eqref{kces} is a strictly decreasing function of $k$ for given $d$ and $n$ (see Appendix \ref{monotonicznosc-D}). For $k=2$ and $k=n$, it obviously recovers the known maximal dimensions of CESs and GESs, respectively, which are explicitly given by
\begin{equation}
    D_{2-CES}^{\max}=d^n-n(d-1) -1
\end{equation}
and
\begin{equation}
    D_{n-CES}^{\max}=(d^{n-1}-1)(d-1).
\end{equation}

\section{Universal construction of $k$-CESs} \label{konstrukcja}

While random subspaces do have significance in quantum information theory (see, e.g., Ref. \cite{Hayden}), it is the explicit constructions, which often help one analyze these objects in more detail, for example, to compute entanglement measures of states supported on them. In this section, we provide a universal analytical construction of $k$-CESs, that is, we introduce a method for building $k$-CESs of any dimension for any combination of the numbers $k$, $n$, and $d$.  With this aim we utilize the approach put forward recently in \cite{DemianowiczQuantum}, which was originally designed for the construction of GESs. It turns out that it is versatile enough to be translated directly to the current case. 

The core idea of the construction is to use a set of product vectors for which the full local spanning of Corollary \ref{corollary} is apparent for properly chosen $r$-partitions of the parties [cf. Eq. \eqref{hilbert-new}]. Consequently, the subspace complementary to the set will be automatically inferred to be void of $(k-1)$-producible states, as needed. The desired properties are held by the Vandermonde vectors:
%
\begin{align}\label{vandermonde-vector}
\ket{v_{d^n}(x)}_{\textbf{A}}:=&\left(1,x,x^2,\dots,x^{d^n-1}\right)^T_{\textbf{A}} \nonumber \\
=&\bigotimes_{m=1}^{n} \left( \sum_{s_m=0}^{d-1} x ^ {q_{m,s_m}}\ket{s_m} \right)_{A_m},
\end{align}
where
\begin{equation}\label{q}
q_{m,s_m}=s_m d^{n-m}.
\end{equation}
Crucially, the Vandermonde matrices, that is matrices with rows being the Vandermonde vectors, 
\begin{equation}
    V_{K,L}
    = \sum_{k=0}^{K-1} \outerp{k}{v_{L}(x_{k})},
\end{equation}
with arbitrary $K$ and $L$,
are totally positive (i.e., all their minors are strictly positive), whenever the nodes $\{x_i\}_i$ obey the relation 
%
\begin{equation}\label{Order}
0<x_0 <x_1 < \dots < x_{K-1}.     
\end{equation}

Let us now consider the following set of linearly independent vectors
\begin{align} \label{nupb-k-ces}
 \left\{ \ket{v_{d^n}(x_i)}_{\textbf{A}}\right\}_{i=0}^{K-1}, 
\end{align}
where the number $K$ satisfies the following inequality 
\begin{equation}
    K \ge t d^{k-1}+ d^{ n-t(k-1)}-t,
\end{equation}
which follows from 
Eq. \eqref{kces}, and 
the nodes $x_i\in\mathbbm{R}$ are ordered as
in Eq. (\ref{Order}).

The claim is that the vectors in Eq. \eqnref{nupb-k-ces} are not extendible with vectors of the entanglement depth less than or equal to $k-1$, 
i.e., the subspace orthogonal to their span is a $k$-CES. The proof is almost immediate. We write
\begin{equation}
  \ket{v_{d^n}(x_i)}_{\textbf{A}} = \ket{\varphi_1^{(i)}}_{S_1}\otimes \ket{\varphi_2^{(i)}}_{S_2} \otimes \ldots \otimes \ket{\varphi_r^{(i)}}_{S_r}
\end{equation}
for an $r$-partition $S_1 | S_2|\dots |S_r$ of the parties corresponding to \eqref{hilbert-new}, in particular \eqref{new-hilbert}. Consider matrices with rows being local vectors on $S_j$'s,
\begin{equation} \label{lokalne-macierze}
    \mathfrak{S}_j= \sum_{k=0}^{K-1} \outerp{k}{\varphi_j^{(k)}}, \quad j=1,2,\dots,r.
\end{equation}

    From the total positivity of Vandermonde matrices with positive nodes (see above) it follows that $\mathfrak{S}_j$'s are also totally positive. This further means that for every $j$ any set of $d^{n_j}$ vectors $\{\ket{\varphi_j^{(i)}}\}_i$ spans the Hilbert space $(\mathbb{C}^d)^{\otimes n_j}$ of $S_j$.
This is the case for any $r$-partition with  $n_i \le k-1$.  
By Corollary \ref{corollary}, we conclude that indeed there does not exist a $(k-1)$-producible
 state in the orthocomplement of the span of vectors \eqref{nupb-k-ces}, i.e., the resulting subspace is a $k$-CES. Its dimension is $d^n-K$, which for the smallest $K$ corresponds to the maximal dimension given in Eq. \eqref{kces}. Explicit (non-orthogonal) bases for such constructed subspaces were given in \cite{DemianowiczQuantum}.

Similar to \cite{DemianowiczQuantum} one could also use discrete Fourier transform matrices in the construction. 
The  method presented here can be viewed as a construction of ''non-orthgonal'' $\mathrm{UPB}_{k-1}$'s (see Section \ref{ces-upb}), whose elements are not mutually orthogonal.

Observe that the construction is flexible in that for a given  $k$-CES we can always increase $k$ by adding some number of Vandermonde vectors with properly chosen nodes to the construction set \eqref{nupb-k-ces}.
For example, consider an eleven-dimensional $2$-CES, or, simply, a CES in the standard terminology, in the system of four qubits. Its construction requires five Vandermonde vectors. If we now add two more vectors to \eqref{nupb-k-ces} we will obtain a nine-dimensional 3-CES  (in Appendix \ref{example-k-ces} we give an example of such a subspace). If we  again add two vectors we will obtain a seven-dimensional 4-CES, i.e. a GES, of the maximal permissible dimension. Further enlargement of the set of Vandermonde vectors leads to GESs of smaller dimensions.

\section{$k$-CES and unextendible product bases}\label{ces-upb}

In this section, we explore the link between $k$-CESs and UPBs.
This line of study is motivated by the important fact that UPBs lead directly to a general construction of multipartite states with undistillable entanglement \cite{upb}.

We are interested in a certain type of UPBs, namely those bases which are unextendible with vectors of the entanglement depth less than $k$ that thus lead to $k$-CESs. We thus propose the following.

\begin{defn}
    A UPB which is unextendible with $(k-1)$-producible vectors, i.e., states with an entanglement depth less than or equal to  $k-1$, is called a $(k-1)$-unextendible product basis, denoted $\mathrm{UPB}_{\mathrm{k-1}}$.
\end{defn}

For $k=2$ we have the standard UPBs (see Section \ref{preliminaries}), 
while for $k=n$--- we have genuinely unextendible product bases (GUPBs), that is UPBs unextendible even with biproduct vectors \cite{upb-to-ges,MD-PRA}, whose existence 
is currently unknown (see Refs. \cite{MD-PRA,gupb-kutrity,kiribela} for recent results related to this topic).

By definition, the orthocomplement of a $\mathrm{UPB}_{\mathrm{k-1}}$ is a $k$-CES.
Importantly, the reverse statement is not necessarily true even in the cases when
the dimension of a $k$-CES is such that  its complement may admit a $\mathrm{UPB}_{\mathrm{k-1}}$. 
In fact, the orthocomplement of a $k$-CES may also be a $k$-CES itself.

We observe that the theoretical minimal cardinality of a $\mathrm{UPB}_{\mathrm{k-1}}$ is [see Eq. \eqref{t} for the definition of $t$]
\begin{widetext}
\begin{equation}\label{m-min-kupb}\small
m^{theor.}_{min}= \begin{cases} d^n & d=2\; \mathrm{and\;} k=n,\\ t(d^{k-1}-1) + d^{n-(k-1)t} & (\mathrm{odd}\; d)\; \mathrm{or}\;[\mathrm{even\;}d\; \mathrm{and}\; \{( k-1|n\; \mathrm{and\;odd}\; t )\; \mathrm{or}\;(k-1 \nmid n\; \mathrm{and\; even\; } t)\}], 
	 \\
t(d^{k-1}-1) + d^{n-(k-1)t}+1 & \mathrm{otherwise}.
\end{cases} 
\end{equation}
\end{widetext}
These numbers correspond to the minimal cardinalities of UPBs \cite{alon-lovasz,feng,min-upb-johnston} for systems defined on Hilbert spaces as in Eq. \eqref{new-hilbert}. We will refer to these bounds as {\it trivial} and discuss their improvement in the next section. 

\subsection{Bound on cardinalities of $\mathrm{UPB}_{\mathrm{k-1}}$}\label{ces-upb-bound}

The trivial bounds from Eq. \eqref{m-min-kupb} can in many cases be strengthened if the internal product structure of partitions is appropriately taken into account.
The improvement can be achieved by a  generalization of  the technique recently introduced in Ref. \cite{MD-PRA} to lower-bound permissible cardinalities of UPB$_{\mathrm{n-1}}$, i.e., GUPBs. There, the pigeonhole principle was the key resource, here, its generalization naturally comes in handy. We have the following.

\begin{prop}\label{prop-bound}
The number of states $m$ in a $\mathrm{UPB}_{\mathrm{k-1}}$ is bounded as follows:
\begin{equation}\label{k-upb-bound}
       m \ge d^{k-1} +  (n-k+1) \left(\left\lfloor \frac{d^{k-1}-2}{k-1} \right\rfloor+1\right).
\end{equation}
\end{prop}
\begin{proof}
Consider a set of mutually orthogonal fully product vectors 
\begin{equation}
\mathcal{K} = \{\ket{v_i}_{\textbf{A}}\}_{i=1}^m, \quad\ket{v_i}_{\textbf{A}}= \bigotimes_{j=1}^{n} \ket{u_j^{(i)}}_{A_j}.
\end{equation}
We will show that if the number of elements in $\mathcal{K}$ is strictly smaller than the right-hand side of \eqref{k-upb-bound} then $\mathcal{K}$ is extendible with vectors of the entanglement depth $k-1$ or smaller, i.e., it is not a  $\mathrm{UPB}_{\mathrm{k-1}}$. 

Any two vectors from $\mathcal{K}$, $\ket{v_p}$ and $\ket{v_q}$, are orthogonal because their local vectors on at least one site $A_k$ are orthogonal, i.e., $\langle u_k^{(p)} | u_{k}^{(q)}\rangle =0$ for at least one value of index $k$.
By the generalized pigeonhole principle \cite{GPP} (see Appendix \ref{gpp}), we conclude that for any vector $\ket{v_i}$ there are $w$ sites such that the total number of vectors orthogonal to $\ket{v_i}$ on these sites is at least
\begin{equation}\label{def-s}
    s:=w \left\lfloor \frac{m-1}{n} \right\rfloor + \min \left(w, m-1-n \left\lfloor \frac{m-1}{n} \right \rfloor\right).
    \displaystyle
\end{equation}
Let us focus on one of the vectors from $\mathcal{K}$, e.g., $\ket{v_1}$. Assume that the said $s$ vectors which are orthogonal to $\ket{v_1}$ are $\mathcal{K}_1=\{\ket{v_2}, \ket{v_3}, \ldots ,\ket{v_{s+1}}\}$ with orthogonality holding on sites from
$\textbf{A}_w:=A_1,A_2,\cdots,A_w$. 
Since there exist a vector orthogonal to $\mathcal{K}_1$ on one of the sites from $\textbf{A}_w$, we immediately infer that 
the elements of $\mathcal{K}_1$ do not span the whole Hilbert space corresponding to 
$\textbf{A}_w$, 
that is
\begin{equation}
\dim\mathrm{span} \left\{\otimes_{j=1}^w \ket{u_j^{(i)}}_{A_j}\right\}_{i=2}^{s+1} < \dim \mathcal{H}_{\textbf{A} _w}.
\end{equation}
 Now, if the remaining vectors $\mathcal{K}_2 := \mathcal{K} \setminus \mathcal{K}_1= \{\ket{v_1}, \ket{v_{s+2}}, \ldots , \ket{v_{m}}\}$  do not span 
 on $\mathbf{A}\setminus\textbf{A}_w$ the corresponding Hilbert space, i.e.,
\begin{align}\label{deficiency}
 \dim\mathrm{span}\: \mathcal{K}_2^{\mathbf{A}\setminus\textbf{A}_w} < \dim \mathcal{H}_{\textbf{A}\setminus\textbf{A} _w} ,
 \end{align}
 where
 \begin{align}
 \mathcal{K}_2^{\mathbf{A}\setminus\textbf{A}_w}=\left\{\otimes_{j={w+1}}^n \ket{u_j^{(i)}}_{A_j}\right\}_{i=1,s+2,s+3,\dots,m},
  \end{align}
 then we will find a vector orthogonal to $\mathcal{K}$ in the form 
\begin{equation}
\left(\otimes_{j=1}^w\ket{u_j^{(1)}}_{A_j}\right) \otimes \ket{\xi}_{\textbf{A}\setminus\textbf{A}_w}
\end{equation}
with an arbitrary $(n-w)$-partite  vector $\ket{\xi}$ orthogonal to $\mathcal{K}_2^{\mathbf{A}\setminus\textbf{A}_w}$.
Clearly, the vector $\ket{\xi}$ would have, in this case, an entanglement depth of at most $(n-w)$,
and as a  consequence, $\mathcal{K}$ would not be a $\mathrm{UPB}_{\mathrm{n-w}}$.

We are interested in the permissible cardinalities of $\mathrm{UPB}_{\mathrm{k-1}}$'s 
in which case we need to consider 
\begin{equation} \label{warek}
    w=n-k+1
\end{equation}
and verify when the condition \eqref{deficiency} holds true. This will provide us with the forbidden cardinalities of the bases and by negating the obtained bound we will arrive at \eqref{k-upb-bound}.
A sufficient condition for \eqref{deficiency} is simply that the number of states in $\mathcal{K}_2^{\mathbf{A}\setminus\textbf{A}_w}$
is smaller than the dimension of $\mathcal{H}_{\textbf{A}\setminus\textbf{A} _w}$, that is
%
\begin{equation}\label{eq}
    m-s \leq d^{n-w}-1 :=r,
\end{equation}
which needs to be solved for $m$.

Let us define the following function corresponding to the left-hand side of Eq. \eqref{eq} 
\begin{align}
    &f_w(m) = m-s \nonumber \\ &= m -w \left\lfloor \frac{m-1}{n} \right\rfloor - \min\left(w, m-1-n \left\lfloor \frac{m-1}{n} \right\rfloor\right).
    \end{align}
    Using 
 the identities $\min(x,y)=x+y-\max(x,y)$ and $\max(x,y)=\max(x-z,y-z)+z$, we can write this function as
    \begin{align}\label{funkcja-ef}
    f_w(m)=&(n-w)\left\lfloor \displaystyle\frac{m-1}{n} \right\rfloor \nonumber \\
    &+ \max\left(m-w-n \left\lfloor \displaystyle\frac{m-1}{n}\right\rfloor,1\right).
\end{align}
It is straightforward to show that this function has the following property (see Appendix \ref{funkcja-fy}) 
\begin{widetext}
\begin{equation} \label{funkcjusz}
    f_w(m+1)= \left\{
    \begin{array}{ccc}
        f_w(m)+1 & \mathrm{if} & m-n \displaystyle\left\lfloor \frac{m}{n} \right\rfloor \ge w+1 \; \mathrm{or} \; n \mid m, \\ [2ex]
        f_w(m) & \mathrm{if} & \hspace{-0.4cm} m-n \displaystyle\left\lfloor \frac{m}{n} \right\rfloor \le w \; \mathrm{and} \; n \nmid m.
    \end{array} 
    \right.
\end{equation}
\end{widetext}
It then follows that for a fixed $w$ function $f_w(m)$  is non-decreasing in $m$ and takes all integer values. This allows us to look for the largest solution of $f_w(m)=r$, that is
\begin{equation} \label{warrunek}
    (n-w)\left\lfloor \frac{m-1}{n} \right\rfloor + \max{\left(m-w-n \left\lfloor \frac{m-1}{n}\right\rfloor,1\right)} = r.
\end{equation}
We now observe that we can replace the maximum with the function inside of it, as for any $m$ for which this function is less than one, we can find another argument $\tilde{m} > m$ that reaches $1$ without changing the value of the whole function under scrutiny; precisely, we can take $\tilde{m}=n \lfloor \frac{m-1}{n} \rfloor +w+1$ for which $f_w(m)=f_w(\tilde{m})$. In turn, we can consider
\begin{equation} \label{warunek}
    (n-w)\left\lfloor \frac{m-1}{n} \right\rfloor + m-w-n \left\lfloor \frac{m-1}{n}\right\rfloor = r
\end{equation}
with the assumption that 
\begin{equation} \label{ograniczenie-floor}
    \left\lfloor \frac{m-1}{n}\right\rfloor \le \frac{m-w-1}{n}.
\end{equation}
Eq. \eqref{warunek} simplifies to
\begin{equation} \label{proste}
    m- w\: q = r+w, \quad \quad q:= \left\lfloor \frac{m-1}{n}\right\rfloor.
\end{equation}
From equations \eqref{ograniczenie-floor} and \eqref{proste} we have
\begin{equation}
q \le \frac{r+qw-1}{n},
\end{equation}
which after solving for $q$ gives
\begin{equation}
   q \le \frac{r-1}{n-w},
\end{equation}
implying further that the largest solution to Eq. \eqref{proste} is
\begin{equation}
    m = r + w + w \left\lfloor \frac{r-1}{n-w} \right\rfloor.
\end{equation}
Plugging $r = d^{n-w}-1$ back in and substituting $w=n-k+1$ [\eqref{warek}], we finally obtain the largest excluded cardinality 
\begin{equation}
    m = d^{k-1} + n-k + (n-k+1) \left\lfloor \frac{d^{k-1}-2}{k-1} \right\rfloor,
\end{equation}
meaning that the minimal permissible size of a $\mathrm{UPB}_{\mathrm{k-1}}$ is
\begin{equation}
    d^{k-1} +  (n-k+1) \left(\left\lfloor \frac{d^{k-1}-2}{k-1} \right\rfloor+1\right),
\end{equation}
as claimed in Eq. \eqref{k-upb-bound}.
\end{proof}

We have proved the statement for equal local dimensions but a similar reasoning can be applied to the general case. We omit the derivation here.

We now show when the obtained bound is non-trivial, i.e., when it gives a  lower bound strictly larger than the one provided by Eq. \eqref{m-min-kupb}. We have the following chain of inequalities:
\begin{align}
&d^{k-1} +  (n-k+1) \left(\left\lfloor \frac{d^{k-1}-2}{k-1} \right\rfloor+1\right) \nonumber\\
& \ge d^{k-1} +(n-k+1) \frac{d^{k-1}-1}{k-1}
\label{first}
\\
&= t(d^{k-1}-1)+1+(d^{k-1}-1)\left(  \frac{n}{k-1} -t \right)\nonumber\\
&\ge t(d^{k-1}-1)+d^{(k-1)(\frac{n}{k-1}-t)} \label{second}\\
&= t(d^{k-1}-1) + d^{n-(k-1)t},
\end{align}
where $t=\lfloor \frac{n}{k-1}\rfloor$.
The first inequality follows from the fact that for any 
pair of integers $x$ and $y>0$ 
 the following inequality holds true:
\begin{equation}
\left\lfloor \frac{x}{y}\right\rfloor = \left\lceil \frac{x+1}{y} \right\rceil -1 \ge \frac{x+1}{y}-1    
\end{equation}
and the equality in (\ref{first}) holds iff $k-1$ divides $d^{k-1}-1$. The second inequality, Eq. (\ref{second}), follows from Bernoullie's inequality, stating that
$(1+x)^r\le 1+xr$ with $x > -1$ and $0\le r \le 1$, here applied with $x=d^{k-1}-1$ and $r=n/(k-1)-t$. The equality in this case holds iff $k-1$ divides $n$ (then $r=0$; $r=1$ is never the case).

Confronting the above conditions for equalities with Eq. \eqref{m-min-kupb} provides us with the cases when Eq. \eqref{k-upb-bound} is certainly non-trivial. For example, this is the case for odd $d$ if additionally it holds that $k-1 \nmid d^{k-1}-1$ or $k-1\nmid n$. 
On the other hand, there are clearly cases when it is trivial. This happens for example in the case of $3$-CESs in systems of four qubits ($n=4$, $d=2$, $k=3$). Then, our bound gives $8$ as the minimal cardinality and the same number is obtained from the theory of bipartite UPBs applied to a $4\otimes 4$ system.  

Note that for $k=n$ (GUPBs) the bound reproduces, as it should, the one from \cite{MD-PRA}, in which case it is always non-trivial.
This bound was recently improved in \cite{kiribela}. At this point, it is not clear whether our current bound on the permissible cardinalities of $\mathrm{UPB}_{\mathrm{k-1}}$'s can be improved using the graph theory techniques of \cite{kiribela}.

\subsection{Construction of $k$-CESs from UPBs}

While the bound of Proposition \ref{prop-bound} puts limitations on the sizes of $\mathrm{UPB_{k-1}}$'s,  it does not say anything about the possibility of the actual constructions of bases with permissible cardinalities. This is particularly important in view of the unknown status of this problem for GUPBs ($k=n$). It turns out that $\mathrm{UPB_{k-1}}$'s do exist for $k<n$ and in fact there are already examples in the literature for systems of four qubits \cite{chen-upb}.

We will now reconstruct one of the four-qubit basis from Ref. \cite{chen-upb} and show the existence of a $\mathrm{UPB_{3}}$ in the case of five qubits. 
In both situations, we are interested in the case of $k=n-1$.
For this aim we use the results of Ref. \cite{UPB-no-quantum}, where a  recursive procedure was given for a construction of UPBs of size $2^{n-1}$ for any $n$.

Let us begin with $n=4$. The UPB from Ref. \cite{UPB-no-quantum} is given by
\begin{align}
\mathcal{K}^{(4)}=\{&|0000\rangle,\quad|01fe\rangle,\quad |1e1e\rangle,\quad |1fe0\rangle ,\nonumber\\
&|e001\rangle,\quad|e1ff\rangle,\quad|fe1f\rangle,\quad|ffe1\rangle \},
\end{align}
where $\{\ket{e},\ket{f}\}$ is any orthonormal basis different from $\{\ket{0},\ket{1}\}$, e.g., $\{\ket{+},\ket{-}\}$. It is easy to see that there exists a vector with an entanglement depth equal to $2$ orthogonal to $\mathcal{K}^{(4)}$, i.e., this set is not a $\mathrm{UPB_{2}}$ . For example, such a vector is given by  $\ket{\varphi}_{A_1A_4}\otimes \ket{\psi}_{A_2A_3}$, where $\ket{\varphi} \perp \{\ket{10}, \ket{f1}\}$ and $\ket{\psi} \perp \{\ket{00}, \ket{1f},\ket{e1}\}$. We can see that this stems from the fact that there are pairs of vectors in $\mathcal{K}^{(4)}$, which are identical on two subsystems. Here, there are three such pairs, but for sets with a cardinality of $8$ this would already be an obstacle if there were only two of them.
However, the following amendment can be made to $\mathcal{K}^{(4)}$ to avoid this problem: 
    in the pairs of vectors with $\ket{1}$ and $\ket{e}$ on the first site ($A_1$), we swap vectors on the third site ($A_3$).
This leads to the following set
\begin{align}
\overline{\mathcal{K}^{(4)}}=\{&|0000\rangle,\quad|01fe\rangle,\quad|1eee\rangle,\quad |1f10\rangle, \nonumber\\
&|e0f1\rangle,\quad|e10f\rangle,\quad|fe1f\rangle,\quad|ffe1\rangle \},
\end{align}
which recovers one of the UPBs given in Ref. \cite{chen-upb}, which showed that it is unextendible for any of the two vs two party bipartitions, i.e., that $\overline{\mathcal{K}^{(4)}}$ is 
 a $\mathrm{UPB_{2}}$. As noted in Section \ref{ces-upb-bound}, its cardinality is the minimal possible.

In the $n=5$ case the UPB from Ref. \cite{UPB-no-quantum} is as follows: 

\begin{align}
    \mathcal{K}^{(5)}= \{
&|00000\rangle, \nonumber \\
&|001fe\rangle, \;|e001f\rangle, \;|fe001\rangle,\;
|1fe00\rangle,\; |01fe0\rangle, \nonumber \\ 
&|01e1e\rangle,\;\; |e01e1\rangle,
\;|1e01e\rangle, \;\;|e1e01\rangle, \;|1e1e0\rangle, \nonumber \\ 
& |1fffe\rangle, \;
|e1fff\rangle,\: |fe1ff\rangle,\; |ffe1f\rangle, |fffe1\rangle \}.  \nonumber
\end{align}
The set has a visibly cyclic structure and when proving the (non)existence of a vector with an entanglement depth of $3$ orthogonal to all its elements one needs to consider only two bipartitions because the rest are equivalent to one of them. These bipartitions are 
$A_1A_2|A_3A_4A_5$ and $A_1A_3|A_2A_4A_5$. With an exhaustive search based on 
Fact 
\ref{crucial} we have verified that $\mathcal{K}^{(5)}$ is already a $\mathrm{UPB}_3$ and no modifications of it are necessary.

Since the UPBs of Ref. \cite{UPB-no-quantum}  were constructed recursively, it is plausible that they could be used, directly or after some modifications similar to those shown above in the four qubit case, to obtain $\mathrm{UPB_{n-2}}$'s for any $n$. 

Let us conclude with an observation that in general  the case $k-1|n$ is somewhat special because then we deal with systems with $t$ subsystems of equal local dimensions $d^{k-1}$ (recall our assumptions of dimension $d$ of each subsystem $A_i$) in the "coarse-grained" Hilbert space [Eq. \eqref{new-hilbert}]. Constructions of UPBs for homogeneous  systems (i.e., with equal local dimensions) are far better explored that those for heterogeneous systems (i.e., different local dimensions). A possible route to find $\mathrm{UPB}_{\mathrm{k-1}}$'s could then be to  build  UPBs first for \eqref{new-hilbert} and then look for their proper "fine-grained" versions in the original space.

\section{Conclusions and outlook}

We proposed to consider the notion of the entanglement depth in the context of entangled subspaces by introducing  the notion of completely entangled subspaces of entanglement depth $k$ ($k$-CESs), that is subspaces composed solely of pure states with an entanglement depth of at least $k$. We presented a universal construction of such subspaces that works for any multipartite systems. We also considered the relationship between unextendible product bases (UPBs) and $k$-CESs. In particular, we provided a non-trivial bound on the cardinalities of UPBs whose orthocomplements are $k$-CESs. Further, we discussed the problem of constructing such UPBs and provided some examples in systems of several qubits.

From the general point of view, it is natural to expect that the  introduced notion may turn out to be relevant for analyses of many-body systems, where the entanglement depth is a figure of merit. 
The already established connection between entangled subspaces and quantum error correction may  suggest that $k$-CESs could be of some use also in this area and their analysis could help us to understand this relationship better. Moreover, it would be interesting to see whether the $k$-CESs we have provided could be useful in further studies on counterexamples to the additivity of the minimum output R\'enyi entropy of quantum channels.
 
There are also several concrete open problems that naturally emerge from our study. For example, it would be desirable to see whether the obtained bound on the cardinalities of UPBs leading to $k$-CESs can be further improved, in particular, using graph-theoretic methods, and to research the possibility of a general practical construction of such UPBs, as this would automatically entail a construction of  PPT states. 
Finally, there is constant demand for practical methods of subspace entanglement certification and quantification tailored to specific types of entangled subspaces and this line or research is also worth exploration.

\section{Acknowledgments}

K.V. and R. A. acknowledge the support of the National Science Center (Poland) through the SONATA BIS project (grant no. 2019/34/E/ST2/00369). Enlightening discussions with M. Wieśniak are acknowledged.

\appendix
\section*{Appendices}

\section{Minimal value of \eqref{do-minimalizacji}}\label{max-dim}

Let us start with  two auxiliary results.

\begin{lem}\label{pierwszy-lemat}
Let $n$ and $k \le n$ be positive integers.
Further, let 
    \begin{equation}
        \boldsymbol{x}=(x_1,x_2, \dots, x_n)
        \end{equation}
        be a non-increasing sequence of non-negative integers, such that 
        \begin{equation}
        x_{i+1} \le x_{i} \le k-1, \quad \quad \sum_i x_i =n
        \end{equation}
        and 
        \begin{equation} \label{def-y}
        \boldsymbol{y}= \Big(\underbrace{k-1,\dots,k-1}_{\text t\; \mathrm{times}},n-t(k-1),0,\dots,0 \Big),
        \end{equation}
        where 
        \begin{equation}
        t= \left\lfloor \frac{n}{k-1}\right\rfloor.
        \end{equation}
    
    Then,  for any $p\le n$ it holds that
    \begin{equation}
       \sum_{i=1}^p x_i \le \sum_{i=1}^p y_i,
    \end{equation}
    i.e.,  $\boldsymbol{x}$ is majorized by $\boldsymbol{y}$, written as $\boldsymbol{x} \prec \boldsymbol{y}$.
\end{lem}
\begin{proof}
The result is obvious.
\end{proof}

\begin{lem}\label{drugi-lemat} (see, e.g., \cite{schur-convex}) Let $g(\cdot)$ be a convex function and $\boldsymbol{x} \prec \boldsymbol{y}$. It then holds that
\begin{equation} \label{major}
\sum_i g(x_i) \le\sum_i g(y_i).
\end{equation}
Furthermore, if $g(\cdot)$ is strictly convex and $\bold{x}$ is not a permutation of $\bold{y}$  then the inequality in Eq. \eqref{major} is strict. 
\end{lem}

We can now prove the following.

\begin{fakt}\label{max-value}
The maximal value of
\begin{equation} \label{to-maximize}
\sum_{i=1}^{n} (d^{n_i} - 1) + 1
\end{equation}
over all  $\boldsymbol{n}=(n_1, n_2, \dots, n_n)$ such that $\sum_i n_i=n$ and $0 \le n_i\le k-1$ is
\begin{equation} \label{maksimum}
t d^{k-1} -t + d^{ n-t(k-1)},
\end{equation}
where $t=\left\lfloor \frac{n}{k-1} \right\rfloor$. 
\end{fakt}

\begin{proof}
   From Lemma \ref{pierwszy-lemat} and Lemma \ref{drugi-lemat} with 
   $g(x)=d^{x}-1$, which is strictly convex, it follows immediately that the maximal value of 
   $\sum_i (d^{n_i}-1)$ over all permissible $\boldsymbol{n}$'s  corresponds to the choice $\boldsymbol{n}= \boldsymbol{y}$ with $\boldsymbol{y}$ defined 
   in Eq. \eqref{def-y}. 
   Eq. \eqref{maksimum} then follows from the direct substitution of the optimal $n_i$'s.
\end{proof}

The minimal value of the expression in Eq. \eqref{do-minimalizacji} is obtained by substracting the obtained value \eqref{maksimum} from $d^n$.

\section{Properties of $D_{k-CES}^{max}$}\label{monotonicznosc-D}
In this appendix we prove some basic properties of the maximal dimension of a $k$-CES. In what follows we denote:
\begin{align}
D(d,k,n)&:= D_{k-CES}^{max}\nonumber\\
&=d^n-\left(t_{n,k} d^{k-1} + d^{T_{n,k}}-t_{n,k}\right)
\end{align}
with
\begin{align}
& t_{n,k}:=t=\left\lfloor \frac{n}{k-1} \right\rfloor, \nonumber \\
& T_{n,k} := n-(k-1)t_{n,k}.
\end{align}
\begin{fakt}
    The maximal dimension of a $k$-CES is strictly
    \begin{itemize}
        \item[(a)]  increasing in $d$ for fixed $k$ and $n$: $D(d+1,k,n)> D(d,k,n)$,
        \item[(b)] decreasing in $k$ for fixed $d$ and $n$: $D(d,k+1,n) < D(d,k,n)$,
        \item[(c)] increasing in $n$ for fixed $k$ and $d$: $D(d,k,n+1)> D(d,k,n)$.
    \end{itemize}\end{fakt}
    \begin{proof}
       \textit{Case (a).} Assume now $d$ is a continuous parameter such that $d \in [2, \infty )$. It holds
       \begin{align}
          & \frac{\partial D (d,k,n)}{ \partial d} =  \nonumber \\
          & = \frac{1}{d} \left[n d^{n}- \left( t_{n,k} (k-1) d^{k-1}+T_{n,k}d^{T_{n,k}}\right) \right]
          \nonumber \\
          & > \frac{n}{d} \left[d^{n}- \left( t_{n,k} d^{k-1}+ \mathrm{sgn} (T_{n,k}) d^{T_{n,k}}\right) \right].
       \end{align}
The term in the parentheses is in fact a difference of the product and the sum of numbers all being larger than or equal to two and as such it is larger than or equal to zero implying that $\partial D /\partial d >0 $. This, in turn, means that $D(d,k,n)$ is strictly increasing in $d$. $\square$
\newline
       \textit{Case (b).} Recall  from Fact \ref{max-value} in  Appendix \ref{max-dim} that the maximal dimensions $D(d,k,n)$ and $D(d,k+1,n)$ correspond in \eqref{to-maximize} to, respectively, $\bold{y}=(k-1,\dots,k-1, T_{n,k})$ with $(k-1)$ appearing $t_{n,k}$ times [cf. Eq. \eqref{def-y}] and 
       \begin{align}
       \tilde{\bold{y}}=(\underbrace{k,\dots,k}_{\text t_{n,k+1}\; \mathrm{times}},T_{n,k+1},0,\dots,0).
       \end{align}
       (notice that $t_{n,k}\ge t_{n,{k+1}}$ and we may need to pad zeros in $\tilde{\bold{y}}$ so that the vectors match in length). Clearly, $\bold{y} \prec \tilde{\bold{y}}$ and by Lemma \ref{drugi-lemat} it holds that $D(d,k+1,n)< D(d,k,n)$. $\square$
\newline
     \textit{  Case (c).} Taking into account that
       \begin{align}
       t_{n+1,k}= \begin{cases} t_{n,k}+1 & \mathrm{if\;\;}  k-1|n+1\\ 
       t_{n,k} & \mathrm{if\;\;}  k-1\nmid n+1
       \end{cases},
       \end{align}
       we have
       \begin{widetext}
       \begin{align}
           D(d,k,n+1)-D(d,k,n)= \begin{cases}
               d^{n+1}-d^n-d^{k-1}+1+d^{T_{n,k}-k+2}(d^{k-2}-1) & \mathrm{if\;\;}  k-1|n+1\\
               (d^n-d^{T_{n,k}})(d-1) & \mathrm{if\;\;}  k-1\nmid n+1
           \end{cases}.
       \end{align}
       \end{widetext}
       Since $2 \le k \le n$ we can easily verify that  $D(d,k,n+1)> D(d,k,n)$. $\square$
    \end{proof}
    Case (c) above supports the claim made in Section \ref{ka-cesy} that the construction of $(n-1)$-CESs in $n$-partite systems from GESs in $(n-1)$-partite ones does not lead to subspaces of the maximal dimension.

\section{Maximal $3$-CES: four qubit case}\label{example-k-ces}

In this appendix, we give a nine-dimensional, i.e., maximal, $3$-CES in the case of four qubits ($n=4$, $d=2$, $k=3$) obtained with the construction in Sec. \ref{konstrukcja}.

With the choice of nodes $x_i=i+1$, $i=0,1,\dots,6$, the Vandermonde vectors \eqref{nupb-k-ces} are now of the form (we omit the transpose and shift the index)
\begin{align}\label{vandy}
    & (1,i,i^2,\dots,i^{15})_{\bf{A}} = \nonumber\\
    & (1,i^8)_{A_1}\otimes (1,i^4)_{A_2}\otimes (1,i^2)_{A_3}\otimes (1,i)_{A_4},
\end{align}
with $i=1,2,\dots,7$. The matrix built from these vectors is totally positive. This implies that matrices constructed from bipartite vectors on any pair $A_m A_n$ [cf. Eq. \ref{lokalne-macierze}] are always full rank, meaning, in particular, that any four ($d^2$ with $d=2$) such vectors  span the whole four dimensional Hilbert space. Corollary \ref{corollary} tells us then that there does not exist a vector of the form $\ket{\phi}_{A_mA_n}\otimes \ket{\varphi}_{\bold{A}\setminus A_mA_n}$, i.e., a $2$-producible vector, which is orthogonal to all the vectors \eqref{vandy}. Consequently, the orthocomplement of the span of the vectors \eqref{vandy} is a $3$-CES. 
\newline
\indent
We present a non-orthogonal basis for this $3$-CES as a matrix whose rows are the basis vectors:

%
\begin{widetext}
%
%

\begin{align}\footnotesize
\left(
\begin{array}{cccccccccccccccc}
 -5040 & 13068 & -13132 & 6769 & -1960 & 322 & -28 & 1 & 0 & 0 & 0 & 0 & 0 & 0 & 0 & 0 \\
 -141120 & 360864 & -354628 & 176400 & -48111 & 7056 & -462 & 0 & 1 & 0 & 0 & 0 & 0 & 0 & 0 & 0 \\
 -2328480 & 5896296 & -5706120 & 2772650 & -729120 & 100653 & -5880 & 0 & 0 & 1 & 0 & 0 & 0 & 0 & 0 & 0 \\
 -29635200 & 74511360 & -71319864 & 34095600 & -8752150 & 1164240 & -63987 & 0 & 0 & 0 & 1 & 0 & 0 & 0 & 0 & 0 \\
 -322494480 & 806546916 & -765765924 & 361808139 & -91318920 & 11851664 & -627396 & 0 & 0 & 0 & 0 & 1 & 0 & 0 & 0 & 0 \\
 -3162075840 & 7876316448 & -7432417356 & 3481077600 & -867888021 & 110702592 & -5715424 & 0 & 0 & 0 & 0 & 0 & 1 & 0 & 0 & 0 \\
 -28805736960 & 71527084992 & -67178631520 & 31255287700 & -7721153440 & 972478507 & -49329280 & 0 & 0 & 0 & 0 & 0 & 0 & 1 & 0 & 0 \\
 -248619571200 & 615829294080 & -576265019968 & 266731264800 & -65430101100 & 8162874720 & -408741333 & 0 & 0 & 0 & 0 & 0 & 0 & 0 & 1 & 0 \\
 -2060056318320 & 5092812168444 & -4751761890876 & 2190505063109 & -534401747880 & 66184608126 & -3281882604 & 0 & 0 & 0 & 0 & 0 & 0 & 0 & 0 & 1 \\
\end{array}
\right).
\end{align}
\end{widetext}

\section{Generalized pigeonhole principle}\label{gpp}
For convenience, we recall here the statement of the generalized pigeonhole principle \cite{GPP}.

\begin{fakt}
If $pq+r$ objects are put into $q$ boxes, then for each $0 \le s \le q$ there exist $s$ boxes with the total number of objects in them at least $ps+\min(r,s)$.
\end{fakt}

\section{Proof of Eq. \eqref{funkcjusz}} \label{funkcja-fy}

Let us start by recalling the definition of the function in question:
    \begin{align}\label{funkcja-fe}
    f_w(m)=&(n-w)\left\lfloor \displaystyle\frac{m-1}{n} \right\rfloor \nonumber \\
    &+ \max\left(m-w-n \left\lfloor \displaystyle\frac{m-1}{n}\right\rfloor,1\right)
\end{align}
with $w=n-k+1$.

\textit{Case (i): $n|m.$} We begin with the case of $n$ being a divisor of $m$. We then have
\begin{align}
\left\lfloor \displaystyle\frac{m}{n} \right\rfloor =\frac{m}{n}  \quad \mathrm{and} \quad \left\lfloor \displaystyle\frac{m-1}{n} \right\rfloor = \frac{m}{n}-1.
\end{align}
It follows that
    \begin{align}\label{drugi}
    f_w(m)&=(n-w) \left(\frac{m}{n}-1 \right) + \max\left(n-w,1\right)  \nonumber\\ 
    &=(n-w)\frac{m}{n}
\end{align}
since $n-w =k-1\ge 1$.
Further
    \begin{align}\label{one}
    f_w(m+1)=&(n-w)\left\lfloor \displaystyle\frac{m}{n} \right\rfloor \nonumber \\
    &+ \max\left(m+1-w-n \left\lfloor \displaystyle\frac{m}{n}\right\rfloor,1\right) \\
    =& (n-w)\frac{m}{n} + \max\left(1-w,1\right) \\
    =& (n-w)\frac{m}{n} +1 =f_w(m)+1
\end{align}
because $w=n-k+1 \ge 1$.

\textit{Case (ii): $n\nmid m$.} We now consider the opposite case of $n$ not being a divisor of $m$. It now holds
\begin{align}\label{rowne}
\left\lfloor \frac{m-1}{n} \right\rfloor =\left\lfloor \frac{m}{n} \right\rfloor.
\end{align}
Within the current case we analyze two subcases related to the value of the function under the maximum in $f_w(m)$. 

\textit{Subcase (ii,a): $m-n \displaystyle\left\lfloor \frac{m}{n} \right\rfloor -w \ge 1$}. We just note that this condition can never be satisfied simultaneously with case (i) and is valid only for $n \nmid m$. We have
\begin{align}
m-n \displaystyle\left\lfloor \frac{m-1}{n} \right\rfloor -w = m-n \displaystyle\left\lfloor \frac{m}{n} \right\rfloor -w \ge 1
\end{align}
in view of  Eq. \eqref{rowne}.
As a consequence, in both $f_w(m)$ and $f_w(m+1)$ we can take the function under the maximum as its value. We then have after trivial simplifications
    \begin{align}
    f_w(m)=&m-w-w\left\lfloor \displaystyle\frac{m-1}{n} \right\rfloor
\end{align}
and
    \begin{align}
    f_w(m+1)=&m-w+1-w\left\lfloor \displaystyle\frac{m}{n} \right\rfloor =f_w(m)+1.
\end{align}

\textit{Subcase (ii,b): $m-n \displaystyle\left\lfloor \frac{m}{n} \right\rfloor -w \le 0$.} Now, the maximum is easily seen to be equal to one for both $f_w(m)$ and $f_w(m+1)$, which, taking into account Eq. \eqref{rowne}, results in this subcase in $f_w(m)=f_w(m+1)$.

This concludes the proof of Eq. \eqref{funkcjusz}.

\end{document}